\documentclass[journal]{IEEEtran}
\usepackage{amssymb}
\usepackage{amsthm}
\usepackage{graphicx}
\usepackage{algorithm}
\usepackage{algorithmic}
\usepackage{cite}
\usepackage{bm}
\usepackage{float}
\usepackage{multirow}
\usepackage{color}
\usepackage{subfigure}
\usepackage{caption}
\usepackage{epstopdf}
\usepackage{hyperref}
\usepackage{fixltx2e}
\usepackage{longtable}
\usepackage{diagbox}
\usepackage{changepage}
\usepackage{amsmath, amsfonts}
\usepackage{dsfont}
\usepackage{makecell}
\usepackage{optidef}
\usepackage{subcaption}
\theoremstyle{definition}

\theoremstyle{remark}
\newtheorem{remark}{Remark}
\theoremstyle{plain}
\newtheorem{theorem}{Theorem}
\newtheorem{lemma}{Lemma}

\newtheorem{corollary}{Corollary}

\begin{document}
\title{Joint Secrecy and Covert Communication (JSACC): An Enhanced Physical Layer Security Approach}

\author{Yanyu Cheng,
Haitao Du,
Liqin Hu,
Wei Yang Bryan Lim,
and~Pan Li,~\IEEEmembership{Fellow,~IEEE}

\thanks{
Yanyu Cheng is with the School of Electrical and Electronic Engineering, Nanyang Technological University, Singapore 639798 (e-mail: ycheng022@e.ntu.edu.sg).

Haitao~Du and Liqin~Hu are with the School of Cyberspace, Hangzhou Dianzi University, Hangzhou 310018, China (e-mail: \{242270031; huliqin\}@hdu.edu.cn).

Wei Yang Bryan Lim is with the College of Computing and Data Science, Nanyang Technological University, Singapore 639798 (e-mail: bryan.limwy@ntu.edu.sg).

Pan Li is with the School of Computer Science, Hangzhou Dianzi University, Hangzhou 310018, China (e-mail: lipan@ieee.org).
}}

\maketitle

\begin{abstract}
In this paper, we propose an enhanced physical layer security approach, named joint secrecy and covert communication (JSACC), which aims to improve the performance of physical layer security (PLS). 
The JSACC system can dynamically switch between secrecy mode and covert mode according to the channel difference between legitimate and illegitimate receivers. 
We further leverage reconfigurable intelligent surface (RIS) to extend the communication range.
For each scenario, we derive the closed-form expressions for the outage probability (OP) and ergodic rate (ER).
To further understand system performance, we derive asymptotic approximations in the high signal-to-noise ratio (SNR) regime to obtain the diversity order and high-SNR slope.
We demonstrate that the diversity order of the JSACC depends on Nakagami fading parameters and the RIS reflecting element number. 
Simulation results are consistent with our theoretical analysis and reveal the superiority of the JSACC system over the conventional secrecy communication (SC) system.
\end{abstract}


\IEEEpeerreviewmaketitle

\section{Introduction}\label{section-1}
With the rapid advancement of wireless communication technologies, physical layer security (PLS) has emerged as a fundamental paradigm for safeguarding data transmission against eavesdropping and malicious attacks by exploiting the inherent characteristics of wireless channels \cite{xiang2020noma}. 
Unlike conventional cryptographic approaches that rely on computational complexity, PLS can provide information-theoretic security guarantees, making it particularly suitable for next-generation wireless systems with stringent latency and reliability requirements. 
PLS primarily focuses on secrecy communication (SC), aiming to ensure the confidentiality of information transmitted from a legitimate transmitter to its intended receiver. 
The broadcast nature of wireless channels, however, makes PLS vulnerable to sophisticated eavesdropping techniques, including signal interception, relay attacks, and eavesdroppers equipped with advanced computational capabilities. These challenges underscore the need for more robust and adaptable security strategies that can operate effectively in complex and dynamic wireless environments, paving the way for a closer examination of SC.

SC provides strong information-theoretic security guarantees by exploiting the advantage of a legitimate receiver's channel over that of potential eavesdroppers. Techniques such as artificial noise (AN) generation, beamforming, and advanced channel coding are commonly employed to enhance the secrecy capacity and prevent unauthorized interception \cite{xiang2019physical}. SC has been widely recognized for its ability to secure critical wireless communications, including those in IoT, autonomous systems, and 5G networks. Nevertheless, its practical effectiveness is heavily dependent on favorable channel conditions, where the main channel consistently outperforms the eavesdropper's channel. In real-world scenarios, especially in dense urban deployments or highly mobile networks, the primary and eavesdropper channels may exhibit comparable or even superior qualities at the eavesdropper's side, causing the achievable secrecy capacity to degrade significantly \cite{gopala2008secrecy, liang2008secure}. Factors such as multipath propagation, rich scattering, and high mobility further exacerbate these limitations, making SC insufficient for guaranteeing confidentiality in many practical settings. These inherent constraints motivate the exploration of alternative paradigms that focus on stealth rather than channel superiority, leading naturally to covert communication (CC) \cite{cheng2025performance}.

CC aims to address the limitations of SC by minimizing the probability of detection by adversaries rather than solely focusing on maximizing secrecy capacity. 
By transmitting signals below the noise floor or leveraging channel uncertainty, CC can significantly reduce the risk of interception, making it particularly attractive for scenarios requiring undetectable communication, such as military networks, vehicular systems, and critical IoT applications \cite{bash2015hiding}. 
However, CC introduces intrinsic tradeoffs, achieving low detection probability often comes at the expense of data rate, communication range, and overall spectral efficiency. 
In dynamic networks, such as mobile ad hoc or vehicular networks, where channel conditions and network topology change rapidly, maintaining both reliability and undetectability becomes even more challenging \cite{shahzad2020covert}. 
Moreover, environmental factors such as interference, fading, and noise variability further complicate the design of effective CC systems. 
These challenges underscore the need for adaptive security mechanisms that can jointly balance covertness, reliability, and efficiency, ensuring secure communication in next-generation wireless networks with diverse and dynamic requirements \cite{cheng2022federated}.

\subsection{Motivation and Contributions} 
SC and CC have their own advantages and limitations \cite{10843324,10923651}. 
SC primarily aims to ensure the confidentiality of transmitted data, protecting it from potential eavesdroppers. 
However, its effectiveness diminishes when the legitimate receiver's channel and the eavesdropper's channel have comparable quality, leading to a reduced secrecy capacity. 
On the other hand, CC focuses on minimizing the probability of detection by adversaries, but often at the cost of communication range and data rate. 
Motivated by their respective shortcomings, we introduce a joint secrecy and covert communication (JSACC) paradigm, designed to combine the strengths of both SC and CC. 
By dynamically switching between secrecy and covert modes based on real-time channel conditions, JSACC offers improved spectral efficiency and enhanced security performance. 
This approach not only optimally allocates resources across varying channel environments but also adapts the system's security and communication performance according to the changing threat models, thereby achieving higher adaptability and robustness in dynamic environments.
This hybrid approach is particularly beneficial in applications where both confidentiality and undetectability are critical, such as in low-altitude economy (LAE) networks and satellite communications \cite{du2022performance}.
The contributions can be summarized as follows.
\begin{itemize}
\item We propose a novel PLS mechanism that optimally switches between SC and CC, based on the channel difference between legitimate and illegitimate receivers.
\item By conducting a comprehensive performance analysis of the JSACC and the conventional SC systems, we derive the closed-form expressions for the outage probability (OP) and ergodic rate (ER) of all cases.
\item To obtain the diversity order and high signal-to-noise ratio (SNR) slope, we further derive asymptotic expressions of the OP and ER in the high-SNR regime.
\item Simulation results show that the JSACC system has lower OP and higher ER than those of the conventional SC system, demonstrating that the performance of the JSACC system is superior to that of the conventional SC system.
\end{itemize}

\subsection{Organization}
The organization of the paper is as follows: 
Section~\ref{section-review} mainly reviews previous work in this field.
Section~\ref{section-2} details the system model and channel assumptions. 
Section~\ref{section-3} elaborates on the secrecy and covert communication mechanisms. 
Sections~\ref{section-4} and \ref{section-er} present the performance analysis of the fixed-rate systems and adaptive-rate systems, respectively. 
Section~\ref {section-5} presents the numerical results. 
Section~\ref{section-6} discusses the conclusions and future directions.

\section{Related Work}\label{section-review}
In this section, we introduce two fundamental technologies in our proposed framework, i.e., secrecy and covert communications.

\subsection{Secrecy Communication}
Under typical circumstances, traditional PLS refers to secrecy communication.
Secrecy communication focuses on ensuring the confidential transmission of information, even in the presence of adversaries trying to intercept or decode the messages. 
According to \cite{gopala2008secrecy}, the authors analyzed secrecy transmission over an ergodic fading channel, proposing a low-complexity power allocation scheme that achieves near-optimal secrecy capacity under partial channel state information (CSI), showing that channel fading enhances secrecy capacity.
In \cite{guan2020intelligent}, the authors analyzed the impact of incorporating AN in an RIS-aided wireless communication system and developed an efficient alternating optimization algorithm for secrecy rate maximization, showing that AN is beneficial, especially under conditions with an increasing number of eavesdroppers close to the RIS.
In \cite{sun2021intelligent}, the authors proposed a joint design of positions and beamforming for a Unmanned Aerial Vehicle (UAV) and RIS-aided network, aiming to maximize the secrecy capacity while exploiting AN to combat eavesdropping, and solved the non-convex problem using alternating optimization and semidefinite relaxation, showing significant performance gains.
In \cite{yang2020secrecy, lv2022ris}, the authors proposed an active RIS-aided secrecy transmission scheme to counteract the double fading effect, optimizing power and secrecy capacity, and showing better energy efficiency than passive RIS in secrecy communication.
In \cite{yao2019secrecy}, the authors analyzed secrecy communication in a UAV-enabled wireless network, where UAVs transmitted confidential information to legitimate receivers while avoiding eavesdroppers, optimizing system configurations to maximize secrecy capacity. 
They introduced a secrecy guard zone technique to protect communication further, demonstrating its effectiveness in enhancing secrecy capacity as UAV and eavesdropper densities increase.
In \cite{ghadi2024physical}, the authors analyzed PLS performance in fluid antenna systems under correlated fading channels, showing that fluid antennas provide more secure transmission than traditional antenna systems.

The evolution of secrecy communication has been significantly advanced by the integration of multiple-input multiple-output (MIMO) technologies and cooperative relay networks. 
In \cite{khisti2010secure}, the authors established fundamental limits for secure communication in MIMO wiretap channels, demonstrating that multiple antennas at both transmitter and receiver can substantially improve secrecy capacity compared to single-antenna systems. The work showed that proper precoding design is crucial for achieving optimal secrecy-reliability tradeoffs.
Cooperative jamming has emerged as another promising approach for enhancing PLS. In \cite{shah2019improving}, the authors investigated cooperative jamming strategies in wireless networks, where legitimate nodes transmit artificial noise to confuse potential eavesdroppers. The study revealed that the strategic placement of cooperative jammers can significantly enhance the network's secrecy capacity.
The deployment of machine learning methods for secrecy communication has garnered considerable interest recently.
In \cite{yang2020deep}, the authors proposed a deep reinforcement learning approach for optimizing beamforming and power allocation in secure wireless communications. Their method adaptively learns optimal transmission strategies without requiring perfect channel state information, achieving superior performance compared to conventional optimization techniques.

Existing research on secrecy communication improves confidentiality through techniques such as AN, RIS-aided transmission, and MIMO beamforming. 
However, its performance relies on the legitimate channel being stronger than the eavesdropper's channel, which is often not guaranteed in mobile or dense environments. 
Our JSACC addresses this limitation by switching to covert communication when secrecy capacity becomes insufficient, ensuring more reliable protection.

\subsection{Covert Communication}
In addition to conventional secure methods, covert communication has emerged as a promising approach to enhance wireless network security due to its inherent high-security level \cite{chen2023covert}.
Covert communication aims to transmit information in a way that is undetectable to adversaries, effectively hiding the communication in plain sight. 
In \cite{sobers2017covert}, the authors analyzed covert communication under the Additive White Gaussian Noise (AWGN) and block fading channels with the help of an uninformed Jammer, showing that Alice can transmit $O(n)$ bits covertly in $n$ channels without decreasing power, even when Willie used an optimal detector.
In \cite{wang2022covert}, the authors presented a covert communication aided by RIS to maximize the transmission rate, demonstrating the effectiveness of the RIS for covert communication.
In \cite{cheng2023performance}, the authors applied covert communication to mobile edge computing (MEC) in sixth-generation networks and showed that there was still much room for future development of covert communication.
In \cite{hu2018covert}, the authors analyzed the potential and performance limits of covert communication within amplify-and-forward one-way relay networks, proposing two transmission schemes, rate-control and power-control, and demonstrating that the relay's forwarding ability enhances the effective covert rate.
In \cite{jiang2024physical}, the authors provide an extensive overview of physical layer covert communication, discussing its fundamental theories, while identifying challenges and future research directions for covert communications in the beyond fifth-generation (B5G) wireless network.

In \cite{zhang2024covert}, the authors analyzed blockchain-based covert communication, proposed a reference model, identified five covert communication patterns, and discussed challenges and future directions for 
secure and cost-effective covert communication using blockchain.
In \cite{li2025covert}, the authors proposed a dual-function RIS for covert communications, optimizing power allocation to improve security performance and reduce phase errors, showing that more RIS elements enhance security.
In \cite{yan2017covert}, the authors investigated covert communication in wireless networks with finite blocklength, analyzing the impact of finite blocklength coding on covert communication performance. The study demonstrated that finite blocklength effects have a significant impact on the achievable covert rate and detection error probability.
In \cite{forouzesh2019covert}, the authors proposed covert communication schemes using artificial noise and beamforming techniques in MIMO systems. The work demonstrated that multiple antennas can be exploited to enhance both covert rate and anti-detection capability by carefully designing transmit beamforming vectors.
In \cite{wu2025covert}, the authors introduced ambient backscatter-assisted covert communication, where the covert transmitter harvests energy from ambient RF signals while simultaneously transmitting covert information. This approach provides natural camouflage for covert communication by mimicking the characteristics of backscatter signals.
In \cite{elsadig2022covert}, the authors applied machine learning techniques to covert communication, proposing deep neural network-based approaches for optimizing covert transmission parameters. The study showed that AI-driven adaptive strategies can significantly improve covert communication performance in dynamic environments.

Existing research on covert communication often overlooks the use of methods such as UAV assistance or machine learning–based adaptation for transmission. 
While effective in reducing detection probability, CC typically suffers from low data rates, limited coverage, and vulnerability to interference and fading. 
JSACC mitigates these drawbacks by activating CC only when secrecy fails, thereby balancing covertness with efficiency and reliability.

\subsection{Combination of SC and CC}
In modern communication systems, the combination of SC and CC offers a robust solution for achieving both confidentiality and undetectability. While SC ensures that the information remains secure from eavesdroppers by maintaining secrecy, CC focuses on hiding the very existence of the communication itself. 
By integrating these two approaches, it is possible not only to protect the transmitted data but also to make the communication undetectable, even in hostile environments. 
This combination is particularly relevant in scenarios where high security is required. 
Leveraging methods such as artificial noise, spread-spectrum techniques, and power control strategies, these systems enhance the resilience of communications, making it harder for adversaries to both intercept and detect transmissions. 

In \cite{wu2020achieving}, the authors presented a secrecy wireless communications paradigm that jointly ensures covertness and secrecy by applying PLS against both detection and eavesdropping attacks.
In \cite{forouzesh2020joint}, the authors investigated joint secrecy and covert communications in the system with two legitimate users, presenting an optimization framework and demonstrating that imperfect CSI significantly degrades performance.
In \cite{wang2019secrecy}, the authors presented a multi-hop relaying strategy to address security and covert communication challenges against UAV surveillance and demonstrate a tradeoff between secrecy/covertness and efficiency, revealing the existence of an optimal number of hops.
In \cite{huang2022distributed}, the authors presented a distributed transmission scheme combining rateless coding and protocol design to achieve both security and covertness against wardens and eavesdroppers and demonstrate that adjusting the covert frame ratio effectively counters inter-collusion attacks while maintaining robustness and low complexity.
In \cite{chen2024enhancing}, the authors proposed a secure covert communication scheme utilizing a multi-antenna transmitter and a UAV jammer, optimizing power to maximize secrecy capacity while ensuring error detection and limiting eavesdropping.
In \cite{lindelauf2009influence}, the authors analyzed optimal communication structures for covert organizations using multi-objective optimization and bargaining game theory, showing that the optimal structure depends on link detection probability and information hierarchy.

Existing works combining SC and CC explore multi-hop relaying, UAV, and distributed coding schemes. 
However, they lack adaptability to dynamic channels.
Our JSACC system overcomes this issue by introducing adaptive switching between SC and CC, maintaining both secrecy and covertness under diverse conditions.

\section{System Model}\label{section-2}
Consider a transmission system consisting of a legitimate transmitter (Alice) and a friendly jammer (Jammer), a receiver (Bob), and an RIS as shown in Fig~\ref{fig-model}.
Bob is a cell-edge user without a direct link to communicate with Alice.
Thereby, Bob needs help from the RIS to communicate with Alice.
A warden (Willie) seeks to determine whether Alice is transmitting to Bob.
Alice controls jammer to hinder Willie from detecting the transmitted information.
All nodes are assumed to have a single antenna, and the RIS consists of $N$ reflecting elements, and the reflection-coefficient matrix is expressed as $\Theta = \mathrm{diag}(\beta _{1}e^{j\theta_{1}},\beta _{2}e^{j\theta_{2}},\dots,\beta _{N}e^{j\theta_{N}} )\quad(j \in \sqrt{-1} )$, in which the $n$th element is $\beta_{n}e^{j\theta_{n}}$, where $\beta_{n} \in \left[0,1\right]$ denotes the amplitude coefficient and $\theta \in  \left [ 0,2\pi \right)$ represents the phase shift coefficient.
Due to the difficulty in guaranteeing the channel quality of the Alice-Bob link, the communication is dynamically switched between SC and CC modes according to the channel difference between the legitimate and illegitimate receivers.
In this paper, we define the switch threshold as $\backepsilon$, where the threshold refers to the indicator for the system to switch from SC to CC or switch from CC to SC.
\begin{figure}
\centering
\includegraphics[width=\linewidth]{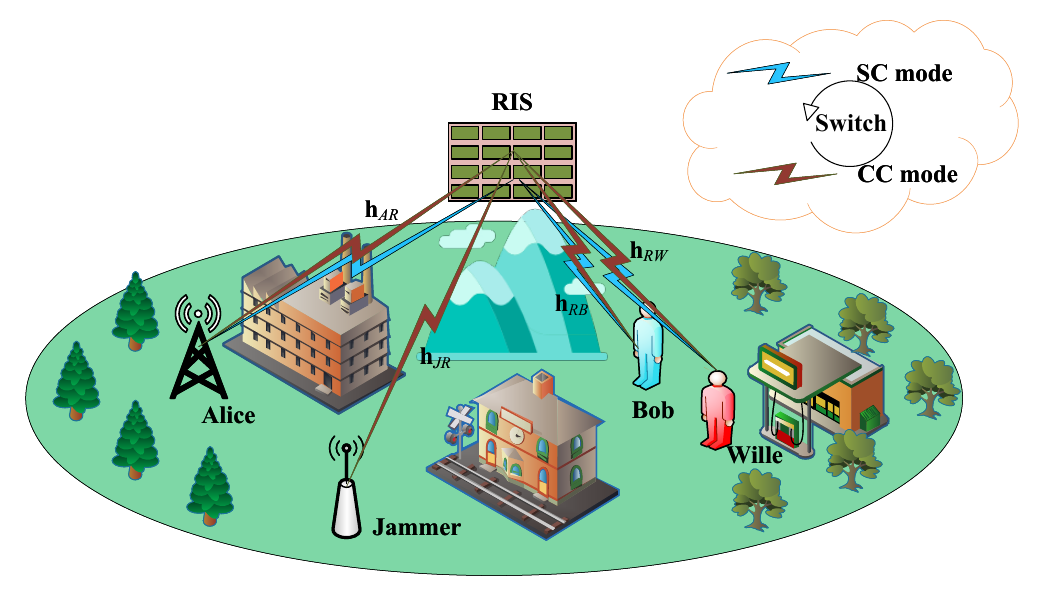}
\caption{The system model of the JSACC, where Alice communicates with Bob in the presence of a warden (Willie) and a friendly jammer (Jammer). When secrecy mode experiences an outage, the system switches to covert mode.}
\label{fig-model}
\end{figure}

\subsection{Channel Model}
We consider all channels to be subject to quasi-static flat fading.
The vectors of the fading channel coefficients between Alice/Jammer and RIS are represented by $\mathbf{h}_{AR}/ \mathbf h_{JR}$. 
The fading channel coefficient vectors between RIS and Bob/Willie are represented by $\mathbf h_{RB}/\mathbf h_{RW}$, and all vectors are $  N\times 1$. 
All the fading coefficients are according to the Nakagami-$m$ fading model with the fading parameter $m_{\chi}$, where $\chi \in \left\{AR, JR, RB, RW \right\}$ indicates the same meaning as in fading coefficients \cite{cheng2021downlink}. 
Specifically, $m_{\chi} = 1$ represents a NLoS scenario, while $m_{\chi} > 1$ corresponds to LoS conditions.
Moreover, each channel experiences path loss with coefficient $\alpha_{\chi}$.

The availability of CSI is characterized under the following assumptions:
1) We assume that Alice perfectly knows all the instantaneous CSI of the links of Alice-RIS-Bob, Alice-RIS-Willie \cite{si2021covert}, Jammer-RIS-Bob, and Jammer-RIS-Willie \cite{8910627}.
3) Willie possesses the perfect CSI of Jammer-RIS-Willie and Alice-RIS-Willie links, which is the worst condition for CC \cite{9524501,9438645}.

Particularly, numerous estimation schemes targeting accurate CSI acquisition have been investigated in the context of RIS-aided systems \cite{he2019cascaded, mishra2019channel}.

\subsection{Signal Model}
Alice transmits the signal sequence $x_{A}(k) = \sqrt{P_{A}}s_{A}(k), k = 1,2,\dots,K$, where $P_{A}$ is Alice's transmit power, $s_{A}(k)\ (\mathbb{E}(\left |s_{A}(k)  \right |^{2} ) = 1)$ is the message transmitted to Bob.
\subsubsection{SC mode}
When the SC rate is greater than the switch threshold $\backepsilon$, the system is activated in SC mode.
The signal at Bob is given by
\begin{equation}\label{eq-Bob-signal-sc}
\begin{split}
y_{B}^{SC}(k)=(\mathbf h_{RB}^{T}\Theta \mathbf  h_{AR} \sqrt{\mathcal{L}_{1}})x_{A}(k) + n_{1}(k),
\end{split}
\end{equation}
where $\mathcal{L}_{1}$ and $\mathcal{L}_{2}$ are the path loss for the links of Alice-RIS-Bob and Jammer-RIS-Bob.
$\mathcal{L}_{1} = L(d_{RB})L(d_{AR}),\mathcal{L}_{2} = L(d_{RW})L(d_{AR})$, where $L(d_{\chi}) = d_{\chi}^{-\alpha_{\chi}}$, $d_{\chi}$ is the distance, $\alpha_{\chi}$ is the loss coefficient, $\chi \in \left \{AR,JR,RB,RW \right \}$.

The signal-to-interference-plus-noise ratio (SINR) of Bob's received signal is expressed as
\begin{equation}\label{eq-Bob-sinr-sc}
\begin{split}
SINR_{B}^{SC} = \frac{P_{t}\left | \mathbf h_{RB}^{T}\Theta \mathbf h_{AR} \right |^{2}\mathcal{L}_{1}}                   {\sigma_{B}^{2}}.
\end{split}
\end{equation}
Similarly, the SINR of Willie's received signal is given by
\begin{equation}\label{eq-Willie-sinr-sc}
\begin{split}
SINR_{W}^{SC} = \frac{P_{t}\left | \mathbf h_{RW}^{T}\Theta \mathbf h_{AR} \right |^{2}\mathcal{L}_{2}}
{\sigma_{W}^{2}} .
\end{split}
\end{equation}
\subsubsection{CC mode}
When the SC rate is less than the switch threshold $\backepsilon$, the system is activated in CC mode, Jammer sends the interference signal sequence $x_{J}(k)=\sqrt{P_{J}}s_{J}(k) \ (\mathbb{E}(\left |s_{J}(k)  \right |^{2} ) = 1)$. 
The Jammer's signal can confuse Willie. 
Since the Jammer's power is a random variable that is uniformly distributed in $\left [ 0, P_{J}^{max} \right ]$. 
Thereby, the Jammer's power probability density function (PDF) is given by
\begin{equation}\label{eq-Jammer-pdf}
\begin{split}
f_{P_{J}}(x) = \left\{\begin{matrix}
\frac{1}{P_{J}^{max}}, & 0\le x \le P_{J}^{max},\\
0, & otherwise.
\end{matrix}\right.
\end{split}
\end{equation}
The signal at Bob is given by
\begin{equation}\label{eq-Bob-signal-cc}
\begin{split}
y_{B}^{CC}(k) = & (\mathbf h_{RB}^{T}\Theta \mathbf h_{AR} \sqrt{\mathcal{L}_{1}})x_{A}(k) \\ & +  (\mathbf h_{RB}^{T}\Theta \mathbf h_{JR} \sqrt{\mathcal{L}_{3}})x_{J}(k) + n_{B}(k),
\end{split}
\end{equation}
where $\mathcal{L}_{3} = L(d_{RB})L(d_{JR})$, $\mathcal{L}_{4} = L(d_{RW})L(d_{JR})$.
The SINR of the signal received by Bob is expressed as
\begin{equation}\label{eq-Bob-sinr-cc}
\begin{split}
SINR_{B}^{CC} = \frac{P_{A}\left | \mathbf h_{RB}^{T}\Theta \mathbf h_{AR} \right |^{2}\mathcal{L}_{1}  }
{\varrho  P_{J}\left | \mathbf h_{RB}^{T}\Theta \mathbf h_{JR} \right |^{2}\mathcal{L}_{3} + \sigma_{B}^{2}},
\end{split}
\end{equation}
where $P_{t} = P_{A} + P_{J}^{max}$, $P_{A} = \kappa P_{t}$ and $\varrho  \in \left [0,1  \right ] $ is the self-interference cancellation coefficient.
In this paper, we assume that $\varrho  = 0$.
Similarly, the SINR of the signal received by Willie is given by:
\begin{equation}\label{eq-Willie-sinr-cc}
\begin{split}
SINR_{W}^{CC} = \frac{P_{A}\left | \mathbf h_{RW}^{T}\Theta \mathbf h_{AR} \right |^{2}\mathcal{L}_{2}}
{P_{J}\left | \mathbf h_{RW}^{T}\Theta \mathbf h_{JR} \right |^{2}\mathcal{L}_{4} + \sigma_{W}^{2}}.
\end{split}
\end{equation}

As an illegitimate warden, Willie seeks to detect whether Alice is transmitting messages to Bob according to the received signal sequence $y_{W}(k),k =1,2,\dots, K$, using the Neyman-Pearson test \cite{zhou2021intelligent,kong2021intelligent}.
Accordingly, Willie faces a binary hypothesis testing problem: 1) the null hypothesis $\mathcal{T}_0$, corresponding to no transmission from Alice to Bob; 2) the alternative hypothesis $\mathcal{T}_1$, Alice is actively transmitting messages to Bob.
Willie's received signals under these two hypotheses are as follows:

\begin{equation}\label{eq-cc-hyp}
\begin{split}
\left\{\begin{matrix}
\mathcal{T}_0 :& \ y_{W}(k) = \mathbf h_{RW}\Theta \mathbf h_{JR} \sqrt{\mathcal{L}_4}x_{J}(k) + n_{W}(k), \\
\mathcal{T}_1 :& \ y_{W}(k) = \mathbf h_{RW}\Theta \mathbf h_{JR} \sqrt{\mathcal{L}_4}x_{J}(k) + \\ &\mathbf h_{RW}\Theta \mathbf h_{AR} \sqrt{\mathcal{L}_2}x_{A}(k) + n_{W}(k),
\end{matrix}\right.
\end{split}
\end{equation}
where $\mathcal{L}_4=L(d_{RW})L(d_{JR})$.

Based on \eqref{eq-cc-hyp}, Willie employs a radiometer to perform the binary detection.
By observing the average power of the signal received at Willie, $P_{W}=\frac{1}{K}\sum_{k=1}^{K}\left| 
y_{W}(k) \right|^2$ the decision rule is defined as follows:
\begin{equation}\label{eq-pw-condition}
\begin{split}
P_W \;\overset{\mathcal U_1}{\underset{\mathcal U_0}{\gtrless}}\;\tau,
\end{split}
\end{equation}
where $\tau > 0$ denotes Willie's detection threshold, with $\mathcal U_0$ and $\mathcal U_1$ corresponding to decisions in favor of $\mathcal{T}_0$ and $\mathcal{T}_1$, respectively.
We assume Willie observes an infinite number of signal samples for the binary detection, i.e., $K \to \infty$ \cite{lv2021covert}.
Consequently, there is certainty of transmitted signals and received noises, and Willie's average received power is given by
\begin{equation}\label{eq-Willie-signal-cc}
\begin{split}
P_W = \left\{\begin{matrix}
\zeta_1 P_J + \zeta_2,	& \mathcal{T}_0, \\
\zeta_1 P_J + \zeta_3,	& \mathcal{T}_1,
\end{matrix}\right.
\end{split}
\end{equation}
where $\zeta_1 = \left | \mathbf h_{RW}^{T}\Theta \mathbf h_{JR} \right |^{2}\mathcal{L}_{4}$, $\zeta_2 = \sigma_{W}^{2}$ and $\zeta_3 = P_{A}\left | \mathbf h_{RW}^{T}\Theta \mathbf h_{AR} \right |^{2}\mathcal{L}_{2} + \sigma_{W}^{2}$.

\section{Performance Metrics and Channel Statistics}\label{section-3}
In this section, we introduce the performance metrics, followed by the channel statistics required for the performance analysis.
\subsection{Performance Metrics}
In this paper, we primarily investigate three metrics to measure system performance: SC rate, CC rate, and detection error probability (DEP).
\subsubsection{Data Rate of SC}
In the SC mode, the Jammer remains inactive, based on \eqref{eq-Bob-sinr-sc} and \eqref{eq-Willie-rate-sc}, the rates of Bob and Willie are
\begin{equation}\label{eq-Bob-rate-sc}
\begin{split}
R_{B}^{SC} = B\log_2 \left( 1 + SINR_{B}^{SC} \right),
\end{split}
\end{equation}
and
\begin{equation}\label{eq-Willie-rate-sc}
\begin{split}
R_{W}^{SC} = B\log_2 \left( 1 + SINR_{W}^{SC} \right).
\end{split}
\end{equation}
Therefore, the secrecy capacity of the SC system is
\begin{equation}\label{eq-rate-sc}
\begin{split}
R_{s} = \left[ R_{B}^{SC} - R_{W}^{SC} \right]^{+}.
\end{split}
\end{equation}
\subsubsection{Data Rate of CC}
In the CC mode, the Jammer remains active, interfering with Willie's receiving ability, based on \eqref{eq-Bob-sinr-cc}, thereby Bob's rate is
\begin{equation}\label{eq-Bob-rate-cc}
\begin{split}
R_{B}^{CC} = B\log_2 \left( 1 + SINR_{B}^{CC} \right).
\end{split}
\end{equation}
\subsubsection{Detection Error Probability}
Willie's detection error indicates that Willie makes an error in binary detection when Willie erroneously considers Alice to be silent while Alice is transmitting, or when Willie incorrectly decides that Alice is transmitting while she is actually silent.
It is defined as
\begin{equation}\label{eq-dep}
\begin{split}
\varsigma = \mathbb{P}_{FA} + \mathbb{P}_{MD},
\end{split}
\end{equation}
where $\mathbb{P}_{MD} = \Pr(\mathcal{T}_1) \Pr\left( \mathcal{U}_0 | \mathcal{T}_1 \right)$ is the miss detection probability, $\mathbb{P}_{FA} = \Pr(\mathcal{T}_0) \Pr\left( \mathcal{U}_1 | \mathcal{T}_0 \right)$ is the false alarm probability \cite{he2017covert}.

\subsection{Channel Statistics}
Accurate channel statistics are essential for analyzing system performance.
\subsubsection{Alice-RIS-Bob Link}
The RIS parameters are configured to optimize the channel quality at Bob, i.e., by maximizing the channel gain of the Alice-RIS-Bob link as follows:
\begin{equation}\label{eq-links-A-R-B-1}
\begin{split}
\left | \mathbf  {h} _{RB}^{T}\Theta \mathbf  {h} _{AR} \right | =\left |     \sum_{n=1}^{N}\beta_{n}h_{RB,n}h_{AR,n}e^{j\theta _{n}}   \right |,
\end{split}
\end{equation}
where $h_{RB,n}$ and $h_{AR,n}$ are the $n$th element of $\mathrm {h} _{RB}$ and $\mathrm {h}_{AR}$, respectively.
By adjusting phase-shift variables $\theta _{n} \left ( n=1,2,\dots, N \right ) $, i.e., maximizing the channel gain, and setting similarly the phases of all $h_{RB,n},h_{AR,n},e^{j\theta _{n}}$. 
Accordingly, the generalized solution is given by $\theta _{n}^{*}=\tilde{\theta } -arg(h_{RB,n},h_{AR,n})$, with $\tilde{\theta }$ denoting a constant of arbitrary value in the interval $\left [ 0,2\pi  \right )$. 
Then, using the optimal $\left \{ \theta _{n}^{*} \right \} $, we have
\begin{equation}\label{eq-links-A-R-B-2}
\begin{split}
\left | \mathbf  {h} _{RB}^{T}\Theta \mathbf  {h} _{AR} \right |^{2} = \beta ^{2}\left ( \sum_{n=1}^{N}\left | h_{RB,n} \right | \left | h_{AR,n} \right |   \right ) ^{2},
\end{split}
\end{equation}
where $\beta_{n}=\beta$, $\forall n$ without loss of generality.

After adopting the optimal $\left \{ \theta _{n}^{*} \right \}$, by referring to \cite[Lemma 1, Lemma 2]{cheng2021downlink}, new channel statistics can be derived by using the CLT as shown in the following lemma.
\begin{lemma}\label{lemma-low-snr}
Denote that $\left | g_{ARB} \right | ^{2}=\frac{\left ( \sum_{n=1}^{N}\left | h_{RB,n} \right | \left | h_{AR,n} \right |  \right )^{2} }{N(1-\mu )}$, where $\mu = \frac{1}{m_{RB}m_{AR}}\left ( \frac{\Gamma (m_{RB}+\frac{1}{2} )}{\Gamma (m_{RB})}  \right ) ^{2} \left ( \frac{\Gamma (m_{AR}+\frac{1}{2} )}{\Gamma (m_{AR})}  \right ) ^{2}$. 
As the number of reflecting elements $N$ becomes large, $\left | g_{ARB} \right | ^{2}$ tends to follow a noncentral chi-square distribution with $\left | g_{ARB} \right | ^{2} \sim   \mathcal{X} _{1}(\lambda )$, where $\lambda = \frac{N\mu}{1-\mu}$.
Its PDF and CDF are given by
\begin{equation}\label{eq-links-A-R-B--low-pdf}
\begin{split}
f _{\left | g_{ARB} \right | ^{2}}(x)=e^{-\frac{x-\lambda }{2} }\sum_{i=0}^{\infty } \frac{\lambda ^{i}x^{i-\frac{1}{2} }}{i!2^{2i+\frac{1}{2}}\Gamma (i+\frac{1}{2} )},
\end{split}
\end{equation}
and
\begin{equation}\label{eq-links-A-R-B-low-cdf}
\begin{split}
F _{\left | g_{ARB} \right | ^{2}}(x)=e^{-\frac{\lambda }{2} }\sum_{i=0}^{\infty } \frac{\lambda ^{i}\gamma (i+\frac{1}{2},\frac{x}{2}  )}{i!2^{i}\Gamma (i+\frac{1}{2} )},
\end{split}
\end{equation}
for $x\ge 0$, respectively.
\end{lemma}

\begin{lemma}\label{lemma-high-snr}
Denote that $Q=\sum_{k=1}^{\infty}|h_{RB,n}||h_{AR,n}|$. 
When $m_{RB}\ne m_{AR}$, the PDF and CDF of $Q$ for $q \to 0^{+}$ are given by
\begin{equation}\label{eq-links-A-R-B--high-pdf}
\begin{split}
f_{Q}^{0^{+}}(q)=\frac{\tilde{m}^{N}}{\Gamma(2m_{s}N)}
q^{2m_{s}N-1}e^{-2\sqrt{m_{s}m_{l}}q},
\end{split}
\end{equation}
and
\begin{equation}\label{eq-links-A-R-B-high-cdf}
\begin{split}
F_{Q}^{0^{+}}(q)=\frac{\tilde{m}^{N}(4m_{s}m_{l})^{-m_{s}N}}{\Gamma(2m_{s}N)}
\gamma(2m_{s}N,2\sqrt{m_{s}m_{l}}q),
\end{split}
\end{equation}
for $q\ge0$, respectively, where $m_{s}=\mathrm{min}\left\{m_{RB},m_{AR}\right\}$, $m_{l}=\mathrm{max}\left\{m_{RB},m_{AR}\right\}$, and $\tilde{m}=\frac{\sqrt{\pi}4^{m_{s}-m_{l}+1}(m_{s}m_{l})^{m_{s}}\Gamma(2m_{s})\Gamma(2m_{l}-2m_{s})}{\Gamma(m_{s})\Gamma(m_{l})\Gamma(m_{l}-m_{s}+\frac{1}{2})}$.
\end{lemma}

\begin{remark}\label{remark-low-snr}
    Since the CLT-based approximation of the PDF $f _{\left | g_{ARB} \right | ^{2}}(x)$ is not valid as $x \to 0^{+}$,
    we use Lemma \ref{lemma-low-snr} addresses the low-SNR regime, while Lemma \ref{lemma-high-snr} deals with the high-SNR regime.
\end{remark}

\subsubsection{Alice-RIS-Willie, Jammer-RIS-Bob, and Jammer-RIS-Willie Links}
For these links, the RIS phase shifts can be treated as random since the RIS is explicitly optimized for the Alice-RIS-Bob link.

Denote that $g_{ARW}=\frac{\mathbf h_{RB}^{T} \Theta \mathbf h_{AR}}{\beta}=\sum_{n=1}^{N}h_{RB,n}h_{AR,n}e^{j\theta_{n}^{*}}$, $g_{JRB}=\frac{\mathbf h_{RB}^{T} \Theta \mathbf h_{JR}}{\beta}=\sum_{n=1}^{N}h_{RB,n}h_{JR,n}e^{j\theta_{n}^{*}}$ and $g_{JRW}=\frac{\mathbf h_{RW}^{T} \Theta \mathbf h_{JR}}{\beta}=\sum_{n=1}^{N}h_{RW,n}h_{JR,n}e^{j\theta_{n}^{*}}$.
When $\left\{\theta_{n} \right\}$ is random and $N$ is large enough, $g_{\phi}\left(\phi \in \left\{ARW, JRB, JRW \right\}\right)$ can be approximated as a complex Gaussian random variable with mean zero and variance $N$, i,e., $g_{\phi} \sim \mathcal{CN}\left(0, N\right)$. 
Hence, its PDF and CDF of $|g_{\phi}|^2$ are given by
\begin{equation}\label{eq-links-other-pdf}
\begin{split}
f _{|g_{\phi }|^{2}}(x)=\frac{1}{N}e^{-\frac{x}{N}},
\end{split}
\end{equation}
and
\begin{equation}\label{eq-links-other-cdf}
\begin{split}
F _{|g_{\phi }|^{2}}(x)=1-e^{-\frac{x}{N}},
\end{split}
\end{equation}
for $x\ge 0$, respectively.

\section{Performance Analysis for Fixed-Rate Communications}\label{section-4}
In this section, we focus on the performance analysis of fixed-rate communications. 
Specifically, we derive the closed-form expressions for OPs and the diversity orders that determine how quickly the OPs decrease with increasing SNR, directly impacting the network's reliability and robustness, which are critical performance metrics in the communications regime.
In this section, the target rate is the switch threshold, which means $\bar{R} = \backepsilon$.

\subsection{Joint Secrecy And Covert Communications}
\subsubsection{Outage Probability}
Firstly, in the JSACC system, we assume the following switching condition
\begin{equation}\label{eq-compare-condition}
\begin{split}
R_{s} \;\overset{\mathcal{G}_{0}}{\underset{\mathcal{G}_{1}}{\gtrless}}\;\bar{R},
\end{split}
\end{equation}
Under each channel condition of $\mathcal{G}_{0}$ and $\mathcal{G}_{1}$, we can derive the corresponding OP.
The events that the system is outaged under $\mathcal{G}_{0}$ and $\mathcal{G}_{1}$ are $\mathcal{H}_{0}$ and $\mathcal{H}_{1}$, respectively.
The OPs under these two conditions are
$\Pr\left(\mathcal{H}_{0}|\mathcal{G}_{0}\right)=0$ and $\Pr\left(\mathcal{H}_{1}|\mathcal{G}_{1}\right)$, respectively.

\begin{theorem}\label{theorem-low-snr}
In the considered JSACC system, the OP of Bob is approximated as
\begin{equation}\label{eq-low-snr-op}
\begin{split}
\mathbb{P}_1 \approx e^{-\frac{\lambda }{2} - \frac{\varpi_2}{N} }\sum_{i=0}^{\infty } \frac{\lambda ^{i}\gamma \left(i+\frac{1}{2},\frac{\varpi_1}{2 \kappa \rho}\right)}{i!2^{i}\Gamma \left(i+\frac{1}{2} \right)} + \sum_{i=1}^{M} \omega_i \Phi_1(t_i),
\end{split}
\end{equation}
where $\rho = \frac{P_{t}}{\sigma^2_B}$, $\varpi_1 = \frac{2^{\frac{\bar{R}}{B}}-1}{\mathcal{L}_{1}\beta^{2}N\left(1-\mu\right)}$, $\varpi_2 =\frac{\varpi_1(1-\kappa)\mathcal{L}_1N(1-\mu)}{\rho \kappa \mathcal{L}_2 2^{\frac{\bar{R}}{B}}}$, $M$ is the number of nodes for the Chebyshev-Gauss quadrature, which determines the accuracy of the approximation, $\omega_i = \frac{\pi}{M}$ is the weight, $t_i = \cos \left( \frac{2i-1}{2M}\pi \right)$, $\Phi_1(t) = \sqrt{1-t^2}\frac{\varpi_2}{2N} \\ e^{-\frac{\varpi_2(t+1)}{2N}}F_{\left | g_{ARB} \right |^{2}} \left( \frac{\varpi_1}{\rho}+ \frac{\varpi_1 (1-\kappa) (t+1)}{2\rho \kappa} \right)$.
\end{theorem}

\begin{proof}\label{proof-low-snr}
First, according to equation \eqref{eq-compare-condition}, $\mathbb{P}_1$ is obtained from
\begin{equation}\label{eq-low-snr-op-proof-1}
\begin{split}
\mathbb{P}_1 &=  \Pr\left(\mathcal{G}_{0}\right)\Pr\left(\mathcal{H}_{0}|\mathcal{G}_{0}\right)+
\Pr\left(\mathcal{G}_{1}\right)\Pr\left(\mathcal{H}_{1}|\mathcal{G}_{1}\right) \\ &
= \Pr\left(SINR_{B}^{CC}<2^{\frac{\bar{R}}{B}}-1,
\frac{1+SINR_{B}^{SC}}{1+SINR_{W}^{SC}}<2^{\frac{\bar{R}}{B}}\right) \\ &
=  \Pr\left( X < \frac{\varpi_1}{\kappa \rho}, X < \frac{\varpi_1}{\rho} + \frac{2^{\frac{\bar{R}}{B}}\mathcal{L}_2y}{\mathcal{L}_{1}N\left(1-\mu\right)} \right),
\end{split}
\end{equation}
where $X = \frac{\left | \mathbf  {h} _{RB}^{T}\Theta \mathbf  {h} _{AR} \right |^{2}N\left(1-\mu\right)}{\beta^2}$ and $Y = \frac{\left |\mathbf h_{RW}^{T}\Theta \mathbf h_{AR} \right |^{2}}{\beta^2}$.
According to \eqref{eq-low-snr-op-proof-1}, the OPs are derived from
\begin{equation}\label{eq-low-snr-op-proof-2}
\begin{split}
\mathbb{P}_1 &=  \Pr \left( X < min\left\{\frac{\varpi_1}{\kappa \rho}, \frac{\varpi_1}{\rho} + \frac{2^{\frac{\bar{R}}{B}}\mathcal{L}_2y}{\mathcal{L}_{1}N\left(1-\mu\right)} \right\} \right)  \\
&= \int_{0}^{\varpi_2} \frac{e^{-\frac{y}{N}}}{N} F_{\left | g_{ARB} \right |^{2}}\left(\frac{\varpi_1}{\rho} + \frac{2^{\frac{\bar{R}}{B}}\mathcal{L}_2y}{\mathcal{L}_{1}N\left(1-\mu\right)} \right) dy \\
&+ \int_{\varpi_2}^{\infty} F_{\left | g_{ARB} \right |^{2}}\left(\frac{\varpi_1}{\kappa \rho} \right) \frac{e^{-\frac{y}{N}}}{N} dy.
\end{split}
\end{equation}
This completes the proof.
\end{proof}
\begin{remark}\label{remark-high-snr}
Remark~\ref{remark-low-snr} indicates that $\mathbb{P}_1$ is accurate in the low-SNR regime but becomes inaccurate in the high-SNR regime. 
We can use Lemma~\ref{lemma-high-snr} to derive the high-SNR approximation of $\mathbb{P}_1$.
\end{remark}

\begin{corollary}\label{proposition-diversity-order}
In the high-SNR regime, when $m_{AR} \neq m_{RB}$, the  approximation of $\mathbb{P}_1$ is given by
\begin{equation}\label{eq-high-snr-op}
\begin{split}
\mathbb{P}_1^{\infty} = \frac{\tilde{m}^{N}\tilde{c}_{1}^{2m_{s}N} \kappa^{-m_sN}}{\Gamma (2m_{s}N +1 )}\rho^{-m_{s}N},
\end{split}
\end{equation}
where $\tilde{c}_{1} = \sqrt{\frac{2^{\frac{\bar{R}}{B}}-1}{\beta^{2}\mathcal{L}_{1}}}$ and $\tilde{c}_2 = \frac{\tilde{c}_1^2(1- \kappa) \mathcal{L}_1}{\kappa 2^{\frac{\bar{R}}{B}}\mathcal{L}_2}\rho^{-1}$.
\end{corollary}

\begin{proof}\label{proof-high-snr}
By expanding the lower incomplete gamma function in \eqref{eq-links-A-R-B-high-cdf}, we have
\begin{equation}\label{eq-high-snr-op-proof-1}
\begin{split}
F_{Q}^{0^+}(q) = \tilde{m}^{N} \sum_{l=0}^{\infty}\frac{(2m_{s}m_{l})^{l}q^{l+2m_{s}N}}{\Gamma(l+2m_{s}N)+1}e^{-2\sqrt{m_{s}m_{l}}q},
\end{split}
\end{equation}
By applying the exponential function's expansion, we have
\begin{equation}\label{eq-high-snr-op-proof-2}
\begin{split}
F_{Q}^{0^{+}}(q)= & \tilde{m}^{N} \sum_{l=0}^{\infty} \frac{\left(2 \sqrt{m_{s} m_{l}}\right)^{l} q^{l+2 m_{s} N}}{\Gamma\left(l+2 m_{s} N+1\right)} \\ & \times \sum_{i=0}^{\infty} \frac{\left(-2 \sqrt{m_{s} m_{l}}\right)^{i} q^{i}}{i!}.
\end{split}
\end{equation}
Meanwhile, the OPs for $\rho \to \infty$ can be derived as
\begin{equation}\label{eq-high-snr-op-proof-3}
\begin{split}
\mathbb{P}_1^{\infty} &= 
\Pr\left(q^2 < \left\{\frac{\tilde{c}_1^2}{\kappa \rho}, \frac{\tilde{c}_1^2}{\rho} + \frac{2^{\frac{\bar{R}}{B}}\mathcal{L}_2Y}{\mathcal{L}_1} \right\} \right)  \\
&=\int_{\tilde{c}_2}^{\infty} \frac{e^{-\frac{y}{N}}}{N} F_{Q}^{0^{+}}(\frac{\tilde{c}_1}{\sqrt{ \kappa \rho}})dy \\ 
&+ \int_{0}^{\tilde{c}_2} \frac{e^{-\frac{y}{N}}}{N} F_{Q}^{0^{+}}(\sqrt{ \frac{\tilde{c}_1^2}{ \rho} + \frac{2^{\frac{\bar{R}}{B}}\mathcal{L}_2Y}{\mathcal{L}_1}})dy,
\end{split}
\end{equation}
In the expansion of $\sqrt{ \frac{\tilde{c}_1^2}{ \rho} + \frac{2^{\frac{\bar{R}}{B}}\mathcal{L}_2Y}{\mathcal{L}_1}}$, the highest order of $\rho$ is zero. 
Finally, by substituting $q_1 = \frac{\tilde{c}_1}{\sqrt{\kappa \rho}}$ into \eqref{eq-high-snr-op-proof-2}, extracting the leading-order term, we can approximate $\mathbb{P}_1^{\infty}$ as \eqref{eq-high-snr-op}.
This completes the proof.
\end{proof}

\begin{corollary}\label{corollary-diversity-order}
In the considered JSACC system, the diversity order of Bob is given by $\mathcal{D}_1 = m_sN$.
\end{corollary}

\begin{proof}\label{proof-diversity-order}
According to Corollary~\ref{proposition-diversity-order}, the diversity order is $m_sN$ when $m_{AR} \neq m_{RB}$. 
When $m_{AR} = m_{RB}$, applying the limit yields the same result, i.e., $m_sN$. 
The proof is complete.
\end{proof}

\begin{remark}\label{remark-3}
The diversity order of Bob in the JSACC system depends on Nakagami fading parameters and the number of RIS reflecting elements.
\end{remark}

\subsubsection{Security Rate}
To measure security performance, we propose a new performance metric $\vartheta_{fixed} = \vartheta_{sc} + \vartheta_{cc}$ called security rate, which is the probability that Willie cannot decode the signal.

\begin{remark}\label{corollary-switch-probability}
For a fixed-rate JSACC system, the switch probability $\mathbb{P}$ is numerically equal to the OP in the conventional SC system, since we set the target rate as the switch threshold.
\end{remark}

First, when the JSACC system is in SC mode, Willie cannot demodulate the transmitted signal.
According to Remark~\ref{corollary-switch-probability}, we can obtain that $\vartheta_{sc} = 1-\mathbb{P}_2$, where $\mathbb{P}_2$ is given by \eqref{eq-tra-low-snr-op}.
Second, we need to obtain the average minimum detection error probability of Willie when the JSACC system is in the CC mode, and Willie has a probability of detecting the transmitted message.
According to \cite{cheng2023performance}, given a detection threshold $\tau$,  the DEP can be characterized by two cases as follows.\\
\textbf{Case 1:} When $\tau_2 < \tau_3$, Willie's DEP is given by
\begin{equation}\label{eq-dep-case-1}
\begin{split}
\varsigma = \left\{\begin{matrix}
1,  & \tau < \tau_1, \\
1-\frac{\tau - \zeta_2}{\zeta_1 P_J^{max}}, & \tau_1 \le \tau < \tau_2, \\
0,  & \tau_2, \le \tau < \tau_3, \\
\frac{\tau-\zeta_3}{\zeta_1-P_J^{max}},   & \tau_3 \le \tau < \tau_4, \\
1,  & \tau \ge \tau_4.
\end{matrix}\right.
\end{split}
\end{equation}
\\
\textbf{Case 2:} When $\tau_2 > \tau_3$, Willie's DEP is given by
\begin{equation}\label{eq-dep-case-2}
\begin{split}
\varsigma = \left\{\begin{matrix}
1,  & \tau < \tau_1, \\
1-\frac{\tau - \zeta_2}{\zeta_1 P_J^{max}}, & \tau_1 \le \tau < \tau_3, \\
1 - \frac{\zeta_3 - \zeta_2}{\zeta_1 P_J^{max}},  & \tau_3, \le \tau < \tau_2, \\
\frac{\tau-\zeta_3}{\zeta_1-P_J^{max}},   & \tau_2 \le \tau < \tau_4, \\
1,  & \tau \ge \tau_4,
\end{matrix}\right.
\end{split}
\end{equation}
where $\tau_1 = \zeta_2$, $\tau_2 = \zeta_1 P_J^{max} + \zeta_2$, $\tau_3 = \zeta_3$, and $\tau_4 = \zeta_1 P_J^{max} + \zeta_3$.

It should be noted that the DEP discussed above is derived under a fixed detection threshold $\tau$.
However, Willie may adjust the threshold to minimize the DEP, representing the worst scenario for CC.
Accordingly, it is essential to evaluate performance under this condition. 
Willie's optimal threshold for minimizing the detection probability is given by
\begin{equation}\label{eq-opt-tau}
\begin{split}
\tau^* = \left\{\begin{matrix}
\left [ \tau_2,\tau_3 \right ],   & \tau_2 \le \tau_3, \\
\left [ \tau_3,\tau_2 \right ],  & \tau_2 > \tau_3.
\end{matrix}\right.
\end{split}
\end{equation}
Using the optimal threshold, 
Willie's minimum DEP is expressed as
\begin{equation}\label{eq-opt-dep}
\begin{split}
\varsigma^* = \left\{\begin{matrix}
0,   & \tau_2 \le \tau_3, \\
1 - \frac{\zeta_3 - \zeta_2}{\zeta_1 P_J^{max}},  & \tau_2 > \tau_3.
\end{matrix}\right.
\end{split}
\end{equation}

\begin{theorem}\label{theorem-amdep}
Bob's security rate in the JSACC system can be given by
\begin{equation}\label{eq-amdep}
\begin{split}
\vartheta_{fixed} = 1 - \frac{\mathbb{P}_2\eta_2 \log\left( 1 + \frac{\eta_1}{\eta_2} \right)}{\eta_1},
\end{split}
\end{equation}
where $\eta_1 = P_J^{max}\mathcal{L}_4$ and $\eta_2 = P_A \mathcal{L}_2$. Specifically, $\mathbb{P}_2$ is the OP of Bob in the conventional SC system.
\end{theorem}

\begin{proof}\label{proof-amdep}
See Appendix \ref{Appen-thro-2}.
\end{proof}

\subsection{The Conventional SC System}
\begin{theorem}\label{theorem-low-snr-tra}
In the conventional SC system, the OP of Bob is approximated as
\begin{equation}\label{eq-tra-low-snr-op}
\begin{split}
\mathbb{P}_2 \approx e^{-\frac{\lambda }{2} }\sum_{k=1}^{K}\omega_k\sum_{i=0}^{\infty } \frac{\lambda ^{i}\gamma (i+\frac{1}{2},\frac{\varpi_1}{2\rho} + \frac{2^{\frac{\bar{R}}{B}-1}\mathcal{L}_2t_{k}}{\mathcal{L}_{1}\left(1-\mu\right)} )}{i!2^{i}\Gamma (i+\frac{1}{2} )},
\end{split}
\end{equation}
where $K$ is the number of nodes for the Gauss-Laguerre quadrature, $t_{k}$ and $\omega_k$ are the roots and corresponding weights of the $k$th-order Laguerre polynomial, respectively.
\end{theorem}

\begin{proof}\label{proof-low-snr-tra}
According to \eqref{eq-rate-sc}, $\mathbb{P}_2$ is given by
\begin{equation}\label{eq-tra-low-snr-op-proof-1}
\begin{split}
\mathbb{P}_2 &=  \Pr\left( R_s < \bar{R} \right) 
= \Pr\left(\frac{1+SINR_{B}^{SC}}{1+SINR_{W}^{SC}}<2^{\frac{\bar{R}}{B}}\right) \\ &
= \int_{0}^{\infty}\int_{0}^{\frac{\varpi_1}{\rho} + \frac{2^{\frac{\bar{R}}{B}}\mathcal{L}_2y}{\mathcal{L}_{1}N\left(1-\mu\right)}} f_{\left | g_{ARB} \right |^{2}}(x)f_{\left | g_{ARW} \right |^{2}}(y)dxdy  \\ &
= \int_{0}^{\infty}F_{\left | g_{ARB} \right |^{2}}\left(\frac{\varpi_1}{\rho} + \frac{2^{\frac{\bar{R}}{B}}\mathcal{L}_2y}{\mathcal{L}_{1}N\left(1-\mu\right)}\right)\frac{e^{-\frac{y}{N}}}{N}dy.
\end{split}
\end{equation}
By substituting $t = \frac{y}{N}$ into \eqref{eq-tra-low-snr-op-proof-1} and according to Gauss-Laguerre, we can approximate $\mathbb{P}_2$ as \eqref{eq-tra-low-snr-op}.
This completes the proof.
\end{proof}

\begin{corollary}\label{proposition-diversity-order-tra}
In the high-SNR regime, the approximation of $\mathbb{P}_2$ are given by
\begin{equation}\label{eq-tra-high-snr-op}
\begin{split}
\mathbb{P}_2^{\infty} = e^{-\frac{\lambda }{2} }\sum_{k=1}^{K}\omega_k\sum_{i=0}^{\infty } \frac{\lambda ^{i}\gamma (i+\frac{1}{2},\frac{2^{\frac{\bar{R}}{B}-1}\mathcal{L}_2t_{k}}{\mathcal{L}_{1}\left(1-\mu\right)} )}{i!2^{i}\Gamma (i+\frac{1}{2} )}.
\end{split}
\end{equation}
\end{corollary}

\begin{proof}\label{proof-high-snr-tra}
By expanding the expression in \eqref{eq-tra-low-snr-op} from point $x = \frac{2^{\frac{\bar{R}}{B}}\mathcal{L}_2t_{k}}{\mathcal{L}_{1}\left(1-\mu\right)}$ using the first-order Taylor expansion, we have
\begin{equation}\label{eq-tra-high-snr-op-proof-1}
\begin{split}
\mathbb{P}_2^{\infty} \approx \sum_{k=1}^{K}\omega_k\Bigg(F_{\left | g_{ARB} \right |^{2}}\left(  \frac{2^{\frac{\bar{R}}{B}}\mathcal{L}_2t_{k}}{\mathcal{L}_{1}\left(1-\mu\right)}\right)\\+\left(\frac{\varpi_1}{\rho}\right)f_{\left | g_{ARB} \right |^{2}}\left(\frac{2^{\frac{\bar{R}}{B}}\mathcal{L}_2t_{k}}{\mathcal{L}_{1}\left(1-\mu\right)}\right)\Bigg),
\end{split}
\end{equation}
Finally, we can approximate $\mathbb{P}_2^{\infty}$ as \eqref{eq-tra-high-snr-op}.
This completes the proof.
\end{proof}

\begin{corollary}\label{corollary-diversity-order-tra}
In the conventional SC system, the diversity order of Bob is given by $\mathcal{D}_2 = 0$.
\end{corollary}

\section{Performance Analysis for Adaptive-Rate Communications}\label{section-er}
In this section, we focus on the performance analysis under adaptive-rate communications. 
Specifically, we derive the closed-form expressions for both ERs and high-SNR slopes of the JSACC and the conventional SC systems, which are critical performance metrics.

\subsection{Joint Secrecy And Covert Communications}
\subsubsection{Ergodic Rate}
The ER of Bob in the JSACC system is presented in the following theorem.
\begin{theorem}\label{theorem-er-jsacc}
In the considered JSACC system, the ER of Bob is given by
\begin{equation}\label{eq-er-jsacc}
\begin{split}
R_1 = \left(1-\mathbb{P}\right)\sum_{i=1}^{K}w_{i}\left(\Phi_2(x_i)-\Phi_3(x_i)\right) + \mathbb{P}\sum_{i=1}^{K}w_{i}\Phi_4(x_i),
\end{split}
\end{equation}
where $\Phi_2(x_i)=\log_{2}\left(1 + \varpi_3 \rho x_i\right) f_{\left | g_{ARB} \right |^{2}}\left(x_i\right)e^{x_i}$, $\Phi_3(x_i)=\log_{2}\left(1 + \varpi_4 \rho x_i\right) f_{\left | g_{ARW} \right |^{2}}\left(x_i\right)e^{x_i}$, $\Phi_4(x_i)=\log_{2}\left(1 + \varpi_3 \kappa \rho x_i\right) f_{\left | g_{ARB} \right |^{2}}\left(x_i\right)e^{x_i}$, $\varpi_3 = \mathcal{L}_1 \beta^2 N (1-\mu)$, $\varpi_4 = \mathcal{L}_2 \beta^2$, and $\mathbb{P} = e^{-\frac{\lambda }{2} }\sum_{k=1}^{K}\omega_k\sum_{i=0}^{\infty } \frac{\lambda ^{i}\gamma (i+\frac{1}{2},\frac{\varpi_1}{2\rho} + \frac{2^{\backepsilon-1}\mathcal{L}_2t_{k}}{\mathcal{L}_{1}\left(1-\mu\right)} )}{i!2^{i}\Gamma (i+\frac{1}{2} )}$.
\end{theorem}

\begin{proof}\label{proof-er-jsacc}
First, according to \eqref{eq-rate-sc}, \eqref{eq-Bob-rate-cc} and Remark \ref{corollary-switch-probability}, we know that $R_1$ is derived as
\begin{equation}\label{eq-er-proof-1}
\begin{split}
R_1 &= \left(1-\mathbb{P}\right)\mathbb E\left( R_s \right)+
\mathbb{P}\mathbb E\left( R_B^{CC} \right) 
\\
&= \left(1-\mathbb{P}\right)\left( \mathbb{E}\left(R_{B}^{SC}\right) - \mathbb{E}\left(R_{W}^{SC}\right) \right) + \mathbb{P}\mathbb E\left( R_B^{CC} \right).
\end{split}
\end{equation}
Then, according to Gauss-Laguerre, we can approximate $\mathbb{E}^{SC}_{B}$ as follows:
\begin{equation}\label{eq-er-sc_b}
\begin{split}
\mathbb{E}\left(R^{SC}_{B}\right) &= \int_{0}^{\infty} \log_{2}\left(1 + \varpi_3 x\right) f_{\left | g_{ARB} \right |^{2}}\left(x\right)dx \\
&= \sum_{i=0}^{\infty}w_{i}\log_{2}\left(1 + \varpi_3 x_i\right) f_{\left | g_{ARB} \right |^{2}}\left(x_i\right)e^{x_i} \\
&=\sum_{i=1}^{K}w_{i}\Phi_2(x_{i}) .
\end{split}
\end{equation}
Similarly, $\mathbb{E}^{SC}_{W}$ and $\mathbb{E}^{CC}_{B}$ can be addressed. 
The proof is complete.
\end{proof}

To further characterize system performance, we define the high-SNR slope as $\mathcal{S}=\lim_{\rho \to \infty} \frac{R(\rho)}{\log_2(\rho)}$.
Its evaluation relies on the asymptotic expression for Bob's ER, which is presented in the following corollary.

\begin{corollary}\label{proposition-er-jsacc}
In the high-SNR regime, the approximation of $R_1$ is given by
\begin{equation}\label{eq-high-snr-er-jsacc}
\begin{split}
R_1^{\infty } \left\{\begin{matrix}
\le \left(1-\mathbb{P}\right) J_1^{\infty} + \mathbb{P} J_2^{upper,\infty}, \\
\ge \left(1-\mathbb{P}\right) J_1^{\infty} + \mathbb{P} J_2^{lower,\infty},
\end{matrix}\right.
\end{split}
\end{equation}
where $\Phi_5(x_i)=\log_{2}\left(\varpi_3x_i\right) f_{\left | g_{ARB} \right |^{2}}\left(x_i\right)e^{x_i}$, $\Phi_6(x_i)=\log_{2}\left(\varpi_4x_i\right) f_{\left | g_{ARW} \right |^{2}}\left(x_i\right)e^{x_i}$, $\Phi_7(x_i)= \frac{f_{\left | g_{ARB} \right |^{2}}\left(x_i\right)e^{x_i}}{x_i}$, $J_1^{\infty} = \sum_{i=1}^{K} w_{i} \Phi_5(x_i) - \sum_{i=1}^{K} w_{i} \Phi_6(x_i)$, $J_2^{upper,\infty} =\log_2\left(1 + \varpi_3\kappa\rho (1+\lambda)\right)$, and $J_2^{lower,\infty} = \log_2\left(1 + \frac{\varpi_3\kappa\rho}{\sum_{i=1}^{K}w_i\Phi_7(x_i)} \right)$.
\end{corollary}
\begin{proof}
See Appendix \ref{Appen-thro-1}.
\end{proof}

\begin{corollary}\label{corollary-er-jsacc}
In the considered JSACC system, the high-SNR slope of Bob is given by $\mathcal{S}_1 = \mathbb{P}$, which means the high-SNR slope is associated with switch probability.
\end{corollary}
\begin{proof}
We have $\mathcal{S}_1 = \frac{dR_1^{\infty}}{d\log_2(\rho)}=(1-\mathbb{P})\frac{dJ_1^{\infty}}{d\log_2(\rho)}+\mathbb{P}\frac{dJ_2^{\infty}}{d\log_2(\rho)}$. Since $J_1^{\infty}$ is a constant, we can obtain $\frac{dJ_1^{\infty}}{d\log_2(\rho)} = 0$.
On the other hand, we know $J_2^{lower,\infty} \le J_2^{\infty} \le J_2^{upper,\infty}$. According to $\frac{dJ_2^{lower,\infty}}{d\log_2(\rho)} = \frac{dJ_2^{upper,\infty}}{d\log_2(\rho)} = 1$, we can obtain $\frac{dJ_2^{\infty}}{d\log_2(\rho)} = 1$.
This completes the proof.
\end{proof}

\subsubsection{Security Rate}
The security rate of Bob under adaptive-rate communications is given by 
\begin{equation}
    \begin{split}
        \vartheta_{adaptive} =  1 - \frac{\mathbb{P}\eta_2 \log\left( 1 + \frac{\eta_1}{\eta_2} \right)}{\eta_1}.
    \end{split}
\end{equation}

\subsection{The Conventional SC System}  
\begin{theorem}\label{theorem-er-benchmark}
In the conventional SC system, the ER of Bob is given by
\begin{equation}\label{eq-er-benchmark}
\begin{split}
R_2 = \sum_{i=1}^{K}w_{i}\left(\Phi_2(x_i)-\Phi_3(x_i)\right).
\end{split}
\end{equation}
\end{theorem}
\begin{proof}
Similar to the proof of Theorem \ref{theorem-er-jsacc}.
\end{proof}

\begin{corollary}\label{proposition-benchmark-er}
In the high-SNR regime, the approximation of $R_2$ is given by
\begin{equation}\label{eq-er-benchmark-high-snr}
\begin{split}
R_2^{\infty} = J_1^{\infty}.
\end{split}
\end{equation}
\end{corollary}

\begin{proof}
Similar to the proof of Corollary~\ref{proposition-er-jsacc}.
\end{proof}

\begin{corollary}\label{corollary-er-sc}
In the conventional SC system, the high-SNR slope of Bob is given by $\mathcal{S}_2 = 0$.
\end{corollary}

\begin{table*}
\centering
\caption{Parameters setting}
\begin{tabular}{|c|c|}
\hline
Bandwidth & $B=1$ MHz\\
\hline
Amplitude-reflection coefficient of RIS & $\beta=0.9$\\
\hline
Distances & $d_{AR}=d_{JR}=100m$ and $d_{RB}=d_{RW}=50m$\\
\hline
Path-loss exponents & $\alpha_{AR}=\alpha_{JR}=3$ and $\alpha_{RB}=\alpha_{RW}=1.5$\\
\hline
Nakagami fading parameters& $m_{AR}=m_{JR}=3$ and $m_{RB}=m_{RW}=1.5$\\
\hline
Target data-rates & $\bar{R} = 0.1\ Mbps$\\
\hline
Number of points for Gauss-Laguerre and Chebyshev-Gauss quadratures& $K=100$ and $M=100$ \\
\hline
\end{tabular}\label{Table-1}
\end{table*}

\section{Numerical Result}\label{section-5}
In this section, we present numerical results to evaluate the performance of the JSACC and the conventional SC systems.
The accuracy of analytical results is further validated through Monte Carlo simulations.
The parameter configurations are detailed in Table~\ref{Table-1} \cite{ding2020simple,ding2020impact}.

\subsection{Average Minimum Detection Error Probability Under the CC Mode}
\begin{figure}
\centering
\includegraphics[width=\linewidth]{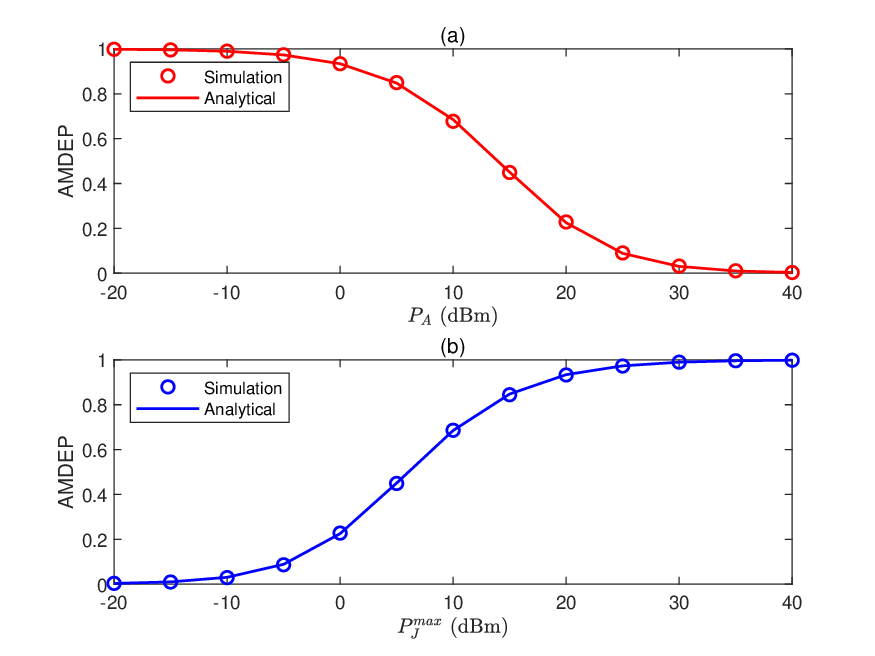}
\caption{AMDEP versus Alice's and Jammer's transmit power.}
\label{fig-amdep}
\end{figure}

\begin{figure}
\centering
\includegraphics[width=\linewidth]{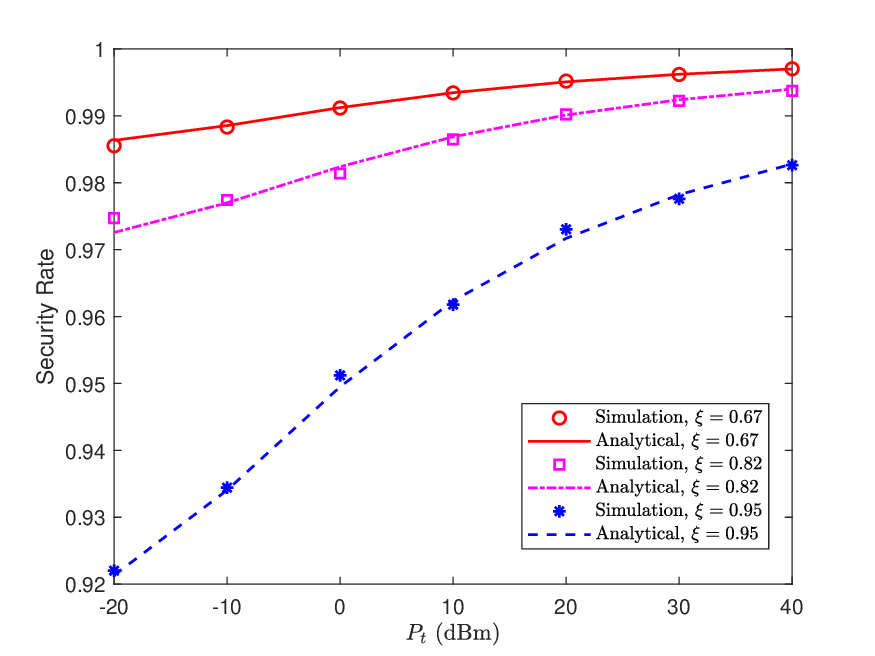}
\caption{Security rate versus $P_t$.}
\label{fig-demo}
\end{figure}

Fig.~\ref{fig-amdep}(a) depicts the AMDEP as a function of Alice's transmit power with the Jammer's power budget fixed at 10 dBm, whereas Fig.~\ref{fig-amdep}(b) presents the AMDEP versus the Jammer's maximum transmit power under a fixed Alice power constraint of 10 dBm.
As observed, AMDEP monotonically decreases with $P_A$ and asymptotically approaches zero as $P_A \to \infty$, while it increases with $P_J^{max}$ and approaches unity as $P_J^{max} \to \infty$.
These phenomena are consistent with intuition, since a larger Jammer's power leads to higher uncertainty for Willie.
In addition, the results indicate that the number of RIS elements has a negligible effect on the AMDEP, highlighting that the principal determinants of covertness performance are the powers transmitted by the Alice and Jammer.

In Fig. \ref{fig-demo}, we plot the security rate versus the transmit power for different power allocation factors $\xi$, where $\xi = \frac{P_A}{P_A + \bar{P_J}}$.
The close correspondence between simulation and analytical results verifies the accuracy of the theoretical model.
It is observed that the security rate gradually increases as $P_t$ increases.
Moreover, a smaller $\xi$ leads to a higher security rate, which means better covertness.
On the other hand, we know that a smaller $\xi$ leads to worse outage performance.
Therefore, the trage-off between outage performance and covertness performance is worthy of investigation in the future.

\subsection{Fixed-Rate Communications}
\begin{figure}
\centering
\includegraphics[width=\linewidth]{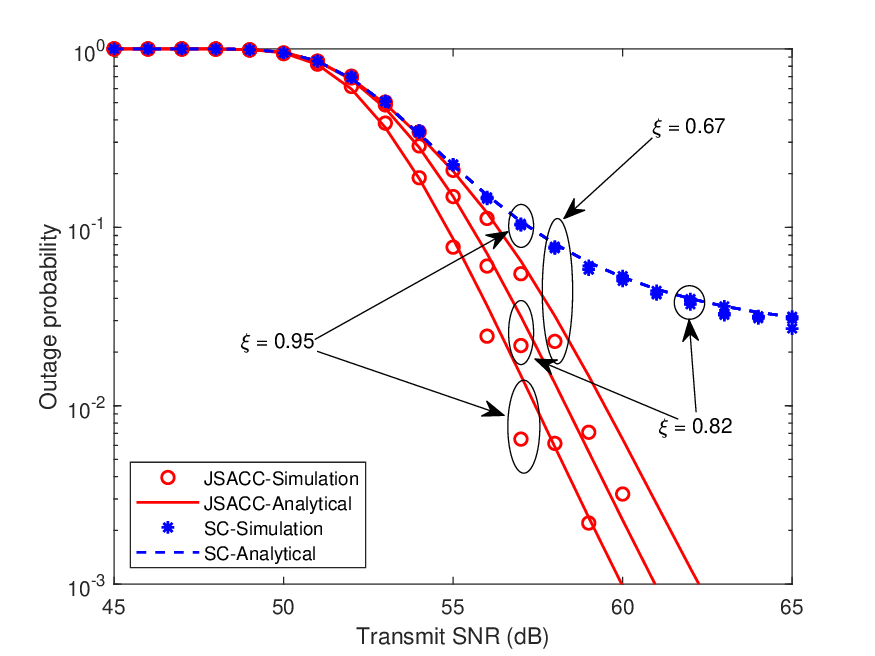}
\caption{OPs derived from the CLT-based channel statistics with $N = 8$.}
\label{fig-low-snr}
\end{figure}

In Fig.~\ref{fig-low-snr}, we plot the OPs versus the transmit SNR in the JSACC and conventional SC systems, where conventional SC is regarded as the benchmark.
Initially, the conventional SC analysis is confirmed to be accurate, since the simulation results are consistent with the analytical expressions given in \eqref{eq-tra-low-snr-op}.
For JSACC, the analysis derived from \eqref{eq-low-snr-op} accurately captures the low-SNR regime. 
Still, it exhibits deviations in the high-SNR regime, which is attributed to the CLT-based approximation of channel statistics.
We observe that the JSACC system consistently outperforms the conventional SC in the low-SNR regime.

\begin{figure}
\centering
\includegraphics[width=\linewidth]{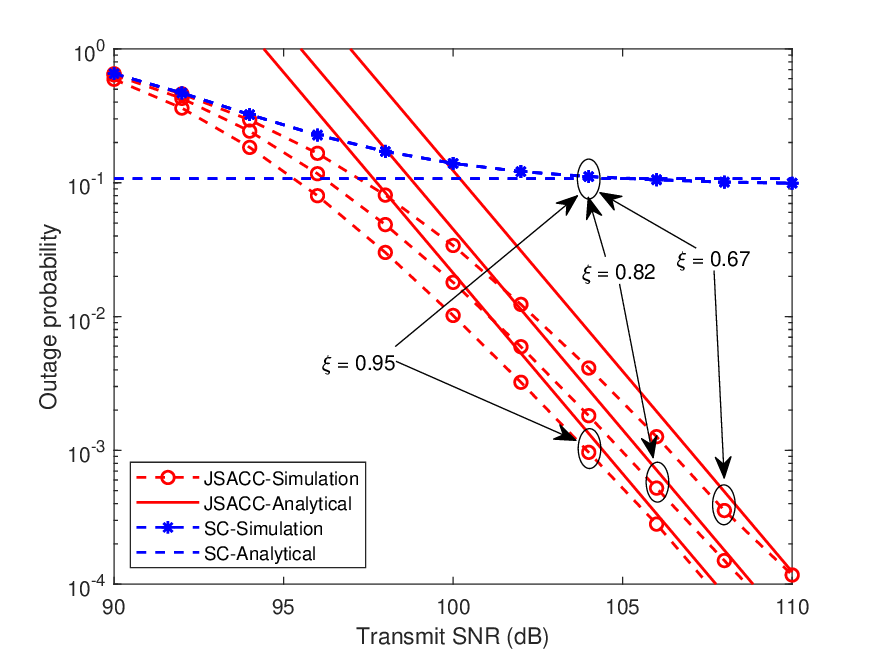}
\caption{High-SNR OPs of Bob in the JSACC and conventional SC when $N = 2$.}
\label{fig-high-snr}
\end{figure}

As the theoretical characterization of the JSACC system becomes less accurate in the high-SNR region shown in Fig~\ref{fig-high-snr}.
As shown, the OPs of both JSACC and conventional SC asymptotically approach the curves derived from \eqref{eq-high-snr-op} and \eqref{eq-tra-high-snr-op}, validating the proposed analysis.
Additionally, by observing slops, we obtain that the diversity order under the JSACC system is $m_sN$ when $N=2$.
For conventional SC, the diversity order equals zero for $N=2$, which is consistent with Corollary~\ref{corollary-diversity-order} and Corollary~\ref{corollary-diversity-order-tra}.
Hence, the results indicate a substantial performance improvement of JSACC over conventional SC.

\subsection{Adaptive-Rate Communications}
We plot the ERs versus the transmit SNR in the JSACC and benchmark SC systems in Fig. \ref{fig-er}, where the conventional SC system is regarded as a benchmark.
Initially, the simulation points for the JSACC and conventional SC systems are seen to align well with the corresponding analytical results derived from \eqref{eq-er-jsacc} and \eqref{eq-er-benchmark}.
It is observed that the ER performance of JSACC is comparable to that of conventional SC at low SNR levels. 
However, as the SNR increases slightly, the performance of JSACC has a higher ER than in the conventional SC, which demonstrates the superiority of JSACC.

Furthermore, the validity of the asymptotic analysis is substantiated through the numerical results illustrated in Fig.~\ref{fig-high-er}.
As shown, the high-SNR approximations derived in \eqref{eq-high-snr-er-jsacc} and \eqref{eq-er-benchmark-high-snr} closely match the simulation results.
In the high-SNR regime, the analysis further reveals that the JSACC system achieves a slope of $\mathbb{P}$. 
In contrast, the conventional SC system attains a slope of $0$, which is consistent with Corollary~\ref{corollary-er-jsacc} and Corollary~\ref{corollary-er-sc}.
Lastly, we obtain that the ER of the JSACC system continues to increase with increasing transmit SNR.
For the conventional SC system, the ER converges to a ceiling, which reveals the advantages of the JSACC system.

\begin{figure}
\centering
\includegraphics[width=\linewidth]{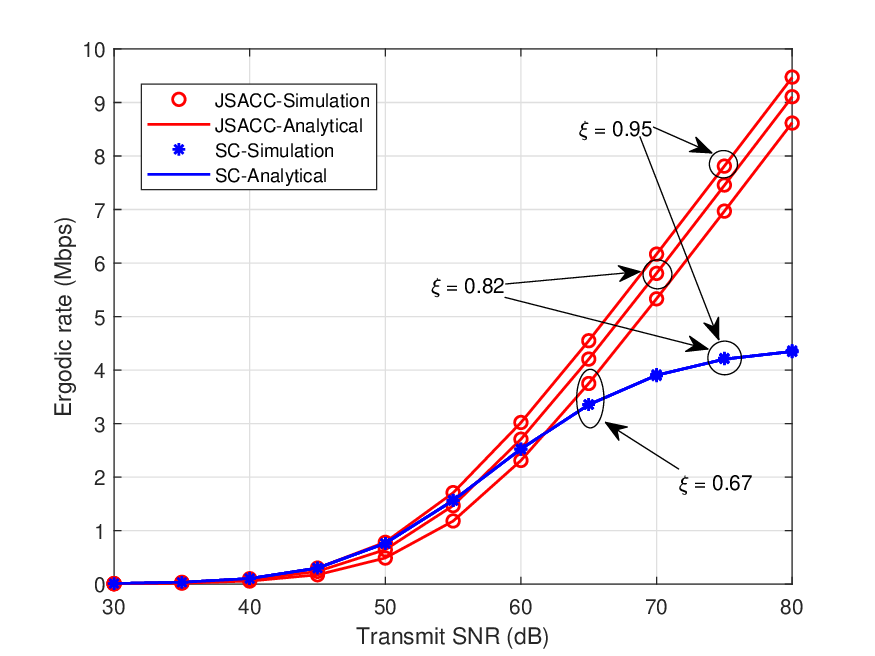}
\caption{ERs versus the transmit SNR in the JSACC and conventional SC systems when $N = 16$.}
\label{fig-er}
\end{figure}

\begin{figure}
\centering
\includegraphics[width=\linewidth]{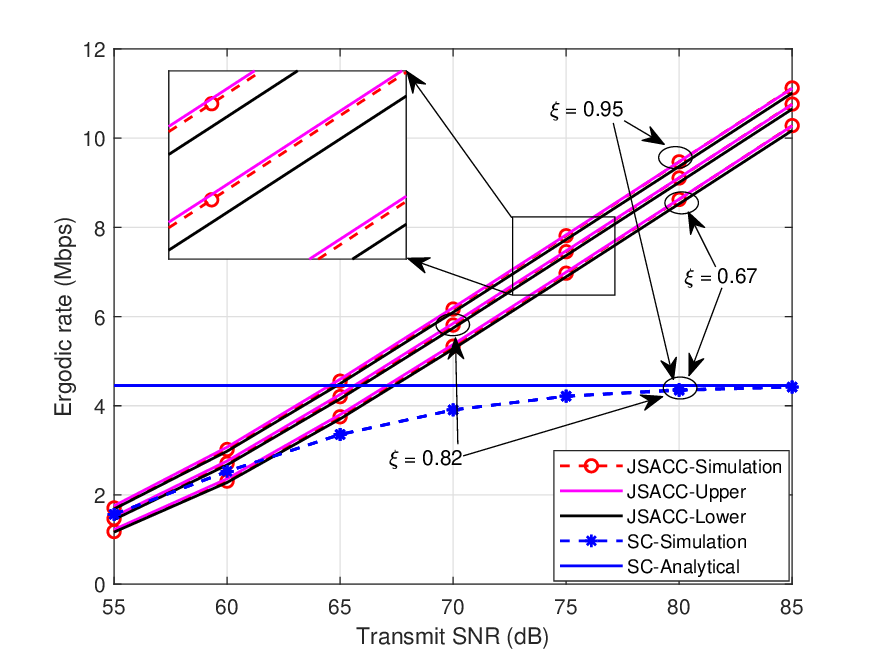}
\caption{High-SNR approximations of ERs.}
\label{fig-high-er}
\end{figure}

\begin{figure}
    \centering
    \includegraphics[width=\linewidth]{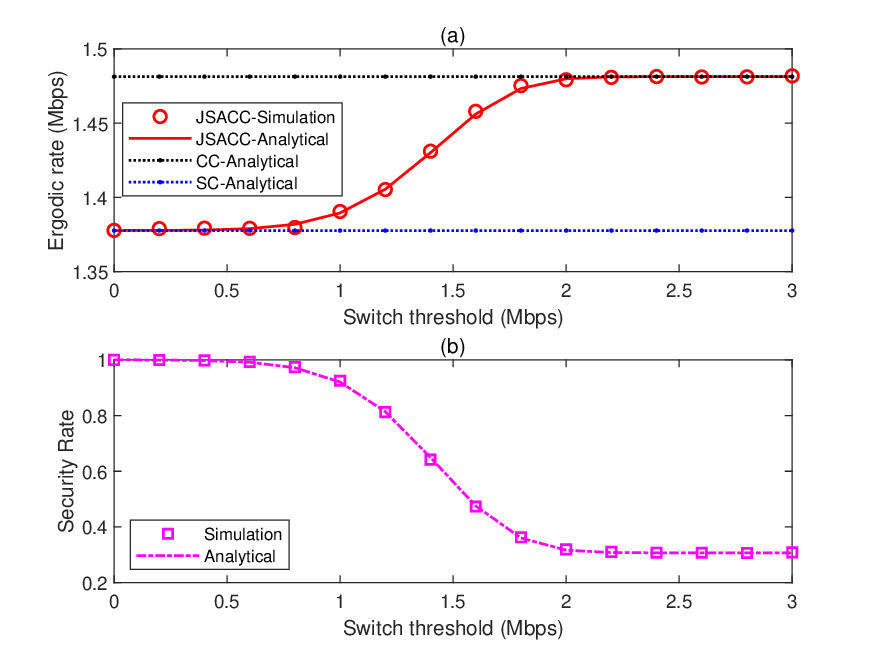}
    \caption{ERs and Security Rate versus switch threshold.}
    \label{fig:placeholder}
\end{figure}

In Fig.~\ref{fig:placeholder}(a), we plot the ERs versus the switch threshold in the JSACC system.
We observe that the ER of JSACC gradually increases as the switch threshold increases.
Furthermore, we plot the security rate versus the switch threshold in the JSACC system in Fig.~\ref{fig:placeholder}(b).
we can observe that the security rate gradually decreases as the switch threshold increases.
It shows that as the switch threshold increases, the capacity performance improves, but the security performance decreases.
The trade-off between capacity performance and security performance is an interesting task for future work.

\section{Conclusion}\label{section-6}
In this paper, we have proposed a new PLS security paradigm named JSACC, which jointly exploits SC and CC. 
By dynamically switching between secrecy and covert modes according to the legitimate user's reception quality, JSACC enhances both spectral efficiency and security performance. 
By deriving the closed-form expressions for OPs and ERs under both scenarios, and analyzing the corresponding diversity orders and high-SNR slopes, we have found that JSACC significantly outperforms the conventional SC system in terms of transmission reliability and secrecy performance.
For future work, adaptive power allocation stands out as a promising approach to further enhance the system's performance.

\begin{appendices}
\section{Proof of Theorem \ref{theorem-amdep}}\label{Appen-thro-2}
\renewcommand{\theequation}{\thesection.\arabic{equation}}
\setcounter{equation}{0}

By leveraging the statistical CSI of all links, the AMDEP of Willie in the CC mode can be derived by
\begin{equation}\label{eq-amdep-proof-1}
\begin{split}
\varsigma^*_c =& \Pr\left( \tau_2 \le \tau_3 \right)\mathbb{E}\left( \varsigma^*|\tau_2 \le \tau_3 \right) \\ & + \Pr\left( \tau_2 > \tau_3 \right)\mathbb{E}\left( \varsigma^*|\tau_2 > \tau_3 \right)
\\ &
= \Pr\left( \tau_2 > \tau_3 \right)\mathbb{E}\left( \varsigma^*|\tau_2 > \tau_3 \right).
\end{split}
\end{equation}
According to \eqref{eq-links-other-pdf} and \eqref{eq-opt-dep}, we can derive by 
\begin{equation}\label{eq-amdep-proof-2}
\begin{split}
\varsigma^*_c = \int_{0}^{\infty} \int_{0}^{\frac{\eta_1 x}{\eta_2}} \frac{e^{\frac{-\left(x+y\right)}{N}}}{N^2}\left(1-\frac{\eta_2 y}{\eta_1 x}\right) dy dx.
\end{split}
\end{equation}
Then, we can derive  $\varsigma^*_c$ by
\begin{equation}\label{eq-amdep-cc}
    \begin{split}
        \varsigma^*_c = 1 - \frac{\eta_2 \log\left( 1 + \frac{\eta_1}{\eta_2} \right)}{\eta_1},
    \end{split}
\end{equation}
Thus, according to \eqref{eq-amdep-cc} and \eqref{eq-tra-low-snr-op}, we obtain AMEDP for the JSACC system is given by
\begin{equation}
    \begin{split}
        \varsigma^*_a = \mathbb{P}_2 \varsigma^*_c,
    \end{split}
\end{equation}
The security rate for the JSACC system can be derived by
\begin{equation}
    \begin{split}
        \vartheta = 1-\mathbb{P}_2 + \mathbb{P}_2 \varsigma^*_a.
    \end{split}
\end{equation}
This completes the proof.

\section{Proof of Corollary \ref{proposition-er-jsacc}}\label{Appen-thro-1}
\renewcommand{\theequation}{\thesection.\arabic{equation}}
\setcounter{equation}{0}

According to \eqref{eq-er-proof-1}, we know that
\begin{equation}\label{eq-high-snr-er-jsacc-proof-1}
\begin{split}
R_1 &= \left(1-\mathbb{P}\right)\underbrace{\int_{0}^{\infty}\int_{0}^{\infty}\log_2\left(\frac{1+\varpi_3x\rho}{1+\varpi_4y\rho}\right)f(x)f(y)dxdy}_{J_1} \\ &+ \mathbb{P}\underbrace{\int_{0}^{\infty}\log_2\left(1+\varpi_3x\kappa\rho\right)f(x)dx}_{J_2}.
\end{split}
\end{equation}
Then, we can analyze $J_1$ and $J_2$.
When $\rho \to \infty$, we can obtain $\lim_{\rho \to \infty} \log_2\left(\frac{1+\varpi_3x\rho}{1+\varpi_4y\rho}\right) \approx \log_2\left(\frac{\varpi_3x}{\varpi_4y}\right)$, $J_1^{\infty}$ can be driven by
\begin{equation}\label{eq-high-snr-er-jsacc-proof-2}
\begin{split}
J_1^{\infty} &=
\int_{0}^{\infty}\int_{0}^{\infty}\log_2\left(\frac{\varpi_3x}{\varpi_4y}\right)f(x)f(y)dxdy
\\& = \int_{0}^{\infty}\log_2\left(\varpi_3x\right)f(x)dx 
\\ & -\int_{0}^{\infty}\log_2\left(\varpi_4y\right)f(y)dy.
\end{split}
\end{equation}
Next, $J_1^{\infty}$ can be estimated through the Gauss-Laguerre quadrature approach, and we have
\begin{equation}\label{eq-high-snr-er-jsacc-proof-3}
\begin{split}
J_1^{\infty} \simeq \sum_{i=1}^{K} w_{i} \Phi_5(x_i) - \sum_{i=1}^{K} w_{i} \Phi_6(x_i).
\end{split}
\end{equation}

On the other hand, by exploiting the properties of the noncentral chi-square distribution, we obtain $\mathbb{E}(x) = 1 + \lambda$.
Then, as $\rho \to \infty$, Jensen's inequality \cite{liu2017non} can be applied to derive an upper bound for $J_2$, we have
\begin{equation}\label{eq-high-snr-er-jsacc-proof-4}
\begin{split}
J_2^{upper,\infty} &= \log_2\left(1 + \varpi_3\kappa\rho \mathbb{E}(x)\right) =\log_2\left(1 + \varpi_3\kappa\rho (1+\lambda)\right) \\
& \ge
\mathbb{E}\left(\log_2\left(1 + \varpi_3\kappa\rho x\right)\right).
\end{split}
\end{equation}
Next, the lower bound of $J_2$ is given by
\begin{equation}\label{eq-high-snr-er-jsacc-proof-5}
\begin{split}
J_2^{lower,\infty} = \log_2\left(1 + \frac{\varpi_3\kappa\rho}{\mathbb{E}(\frac{1}{x})} \right)
\le
\mathbb{E}\left(\log_2\left(1 + \varpi_3\kappa\rho x\right)\right).
\end{split}
\end{equation}
$\mathbb{E}(\frac{1}{x})$ can similarly be approximated using the Gauss-Laguerre quadrature, we have $\mathbb{E}(\frac{1}{x}) = \sum_{i=1}^{K}w_i\Phi_7(x_i)$.
This completes the proof.
\end{appendices}

\bibliographystyle{IEEEtran}
\bibliography{ref.bib}

\end{document}